\pdfoutput=1
\documentclass[sigconf]{acmart}
\makeatletter
\gdef\@copyrightpermission{
 \begin{minipage}{0.3\columnwidth}
  \href{https://creativecommons.org/licenses/by/4.0/}{\includegraphics[width=0.90\textwidth]{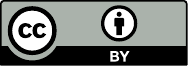}}
 \end{minipage}\hfill
 \begin{minipage}{0.7\columnwidth}
  \href{https://creativecommons.org/licenses/by/4.0/}{This work is licensed under a Creative Commons Attribution International 4.0 License.}
 \end{minipage}
 \vspace{5pt}
}
\def\@ACM@checkaffil{
    \if@ACM@instpresent\else
    \ClassWarningNoLine{\@classname}{No institution present for an affiliation}%
    \fi
    \if@ACM@citypresent\else
    \ClassWarningNoLine{\@classname}{No city present for an affiliation}%
    \fi
    \if@ACM@countrypresent\else
        \ClassWarningNoLine{\@classname}{No country present for an affiliation}%
    \fi
}
\usepackage{enumitem}
\usepackage{threeparttable}
\usepackage{multirow}
\usepackage{bm}
\usepackage{subfigure}
\usepackage{textcomp}%
\usepackage{changepage}

\usepackage[justification=justified,skip=2pt]{caption}

\setlist[itemize]{topsep=2pt}

\makeatletter
\g@addto@macro\normalsize{%
  \abovedisplayskip 3pt plus 2pt minus 3pt%
  \belowdisplayskip \abovedisplayskip
  \abovedisplayshortskip 3pt plus2pt  minus3pt%
  \belowdisplayshortskip 3pt plus2pt minus3pt%
}

\AtBeginDocument{%
  \providecommand\BibTeX{{%
    \normalfont B\kern-0.5em{\scshape i\kern-0.25em b}\kern-0.8em\TeX}}}

\copyrightyear{2024} 
\acmYear{2024} 
\setcopyright{rightsretained} 
\acmConference[SIGIR '24]{Proceedings of the 47th International ACM SIGIR Conference on Research and Development in Information Retrieval}{July 14--18, 2024}{Washington, DC, USA}
\acmBooktitle{Proceedings of the 47th International ACM SIGIR Conference on Research and Development in Information Retrieval (SIGIR '24), July 14--18, 2024, Washington, DC, USA}
\acmDOI{10.1145/3626772.3657684}
\acmISBN{979-8-4007-0431-4/24/07}

\begin{document}

\title{LoRec: Combating Poisons with Large Language Model for Robust Sequential Recommendation}

\author{Kaike Zhang}
\affiliation{%
  \institution{CAS Key Laboratory of AI Safety, Institute of Computing Technology, Chinese Academy of Sciences}
  \country{ }
}
\affiliation{%
  \institution{University of Chinese Academy}
  \country{of Sciences, Beijing, China}
}
\email{zhangkaike21s@ict.ac.cn}

\author{Qi Cao}
\affiliation{%
  \institution{CAS Key Laboratory of AI Safety, Institute of Computing Technology, Chinese Academy of Sciences,}
  \country{Beijing, China}
}
\email{caoqi@ict.ac.cn}
\authornote{Corresponding author}

\author{Yunfan Wu}
\affiliation{%
  \institution{CAS Key Laboratory of AI Safety, Institute of Computing Technology, Chinese Academy of Sciences}
  \country{ }
}
\affiliation{%
  \institution{University of Chinese Academy}
  \country{of Sciences, Beijing, China}
}
\email{wuyunfan19b@ict.ac.cn}

\author{Fei Sun}
\affiliation{%
  \institution{CAS Key Laboratory of AI Safety, Institute of Computing Technology, Chinese Academy of Sciences,}
  \country{Beijing, China}
}
\email{sunfei@ict.ac.cn}

\author{Huawei Shen}
\affiliation{%
  \institution{CAS Key Laboratory of AI Safety, Institute of Computing Technology, Chinese Academy of Sciences,}
  \country{Beijing, China}
}
\email{shenhuawei@ict.ac.cn}

\author{Xueqi Cheng}
\affiliation{%
  \institution{CAS Key Laboratory of AI Safety, Institute of Computing Technology, Chinese Academy of Sciences,}
  \country{Beijing, China}
}
\email{cxq@ict.ac.cn}

\renewcommand{\shortauthors}{Kaike Zhang et al.}

\begin{abstract}

Sequential recommender systems stand out for their ability to capture users' dynamic interests and the patterns of item transitions.
However, the inherent openness of sequential recommender systems renders them vulnerable to poisoning attacks, where fraudsters are injected into the training data to manipulate learned patterns. Traditional defense methods predominantly depend on predefined assumptions or rules extracted from specific known attacks, limiting their generalizability to unknown attacks.
To solve the above problems, considering the rich open-world knowledge encapsulated in Large Language Models (LLMs), we attempt to introduce LLMs into defense methods to broaden the knowledge beyond limited known attacks. We propose \textbf{LoRec}, an innovative framework that employs \textbf{L}LM-Enhanced Calibration to strengthen the r\textbf{\Large o}bustness of sequential \textbf{Rec}ommender systems against poisoning attacks. LoRec integrates an LLM-enhanced CalibraTor (LCT) that refines the training process of sequential recommender systems with knowledge derived from LLMs, applying a user-wise reweighting to diminish the impact of attacks. Incorporating LLMs' open-world knowledge, the LCT effectively converts the limited, specific priors or rules into a more general pattern of fraudsters, offering improved defenses against poisons. Our comprehensive experiments validate that LoRec, as a general framework, significantly strengthens the robustness of sequential recommender systems.

\end{abstract}



\maketitle

\section{INTRODUCTION}
\begin{figure}
    \centering
    \subfigure[Defense Methods on Sequential Recommender Systems against Attacks]{
    \includegraphics[width=0.475\textwidth]{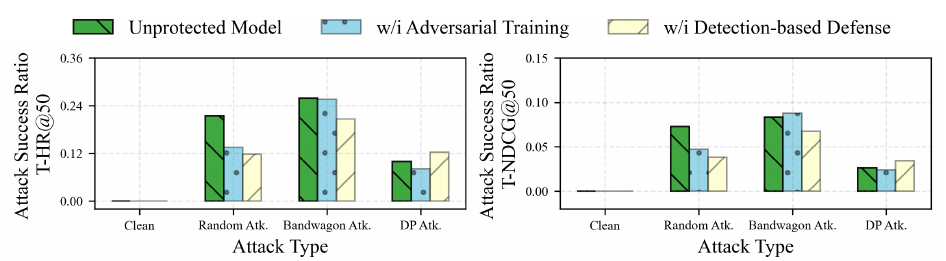}
    \label{fig:intro_attack}
    }
    \subfigure[LLM Knowledge for General Patterns of Fraudsters in Defense]{
    \includegraphics[width=0.475\textwidth]{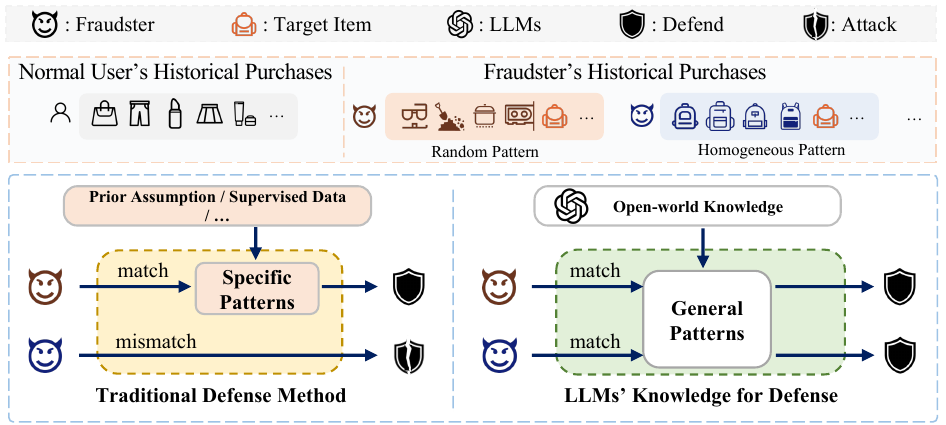}
    \label{fig:llm_know}
    }
    \caption{(a) Sequential recommender systems show vulnerability to attacks despite various defense strategies. (b) Utilizing LLMs' open-world knowledge to enhance defense by generalizing from specific to general patterns.}
\end{figure}

Sequential recommender systems are increasingly popular for personalized recommendations by capturing the dynamic interests of users and the evolving transition patterns of items~\cite{kang2018selfattentive, sun2019bert4rec, zhou2022filterenhanced}. However, the inherent openness of these recommender systems allows attackers to effortlessly inject fraudsters to manipulate the learned item transition patterns, thus fulfilling objectives such as promoting target items, also known as poisoning attacks~\cite{huang2021data, tang2020revisiting}. Such manipulations can drastically skew the distribution of target item exposure, damaging user experience and hindering the long-term development of recommender systems~\cite{zhang2023robust}.

Defensive methods against poisoning attacks within recommender systems typically fall into two groups~\cite{zhang2023robust}, i.e., constructing robust models through (1) adversarial training~\cite{he2018adversarial, li2020adversarial, wu2021fight, yue2022defending,ye2023towards}, or (2) detecting and eliminating injected fraudsters~\cite{chung2013beta, zhang2014hhtsvm, yang2016rescale, zhang2020gcnbased, liu2020recommending}. 
Adversarial training typically follows a ``min-max'' paradigm, which finds attacks that maximally damage performance and then trains the optimal model to minimize the influence of such attacks. However, this ``min-max'' paradigm remains a significant gap with real-world attacks, such as target item promotion~\cite{tang2020revisiting}, failing to show satisfiable defensibility. As shown in Figure~\ref{fig:intro_attack}, adversarial training is less effective against Bandwagon~\cite{mobasher2007toward} and DP attacks~\cite{huang2021data}.

Unlike adversarial training, detection-based methods scan the entire dataset to identify potential fraudsters according to human priors or parts of known fraudster patterns.
Subsequent steps include either the removal of these identified users (hard detection)~\cite{zhang2014hhtsvm, yang2016rescale} or diminishing their impact during training by weight adjustment (soft detection)~\cite{zhang2020gcnbased}.
Unfortunately, the performance of these methods is limited by their specific knowledge from heuristic priors or supervised fraudsters~\cite{chung2013beta, zhang2020gcnbased, yang2016rescale}. As shown in Figure~\ref{fig:intro_attack}, detection-based methods utilizing Bandwagon attacks as supervised data exhibit enhanced defense against such attacks, but fail to defend against DP attacks which deviate from the known supervised fraudsters.
In practical scenarios, as attacks evolve, the ability of defense methods to adapt to unknown attacks is crucial. However, these methods can only mitigate impacts aligning with their specific knowledge and lack broader generalization to unknown attacks.

Recently, Large Language Models (LLMs) have achieved significant advancements across a multitude of fields, showing extraordinary capabilities in diverse applications~\cite{jiao2023chatgpt, li2023empowering}.
Researchers further demonstrate LLMs' ability to encapsulate expansive open-world knowledge, which can aid various tasks in achieving improved generalizability~\cite{fan2023recommender, hou2023learning}. To address the aforementioned challenges, particularly the limited generalizability of existing defense methods against unknown attacks (as illustrated in the left part of Figure~\ref{fig:llm_know}), our work investigates leveraging the open-world knowledge of LLMs to enhance defense methods. By utilizing LLMs, we aim to move beyond the specific attack patterns learned from predetermined assumptions or supervised data towards a broader understanding of fraudsters for developing robust sequential recommender systems, as shown in the right part of Figure~\ref{fig:llm_know}.
 
To achieve the above goal, we present \textbf{LoRec}, an innovative framework that leverages an \textbf{L}LM-Enhanced CalibraTor (LCT) to bolster the r\textbf{\Large o}bustness of sequential \textbf{Rec}ommender systems against poisoning attacks. The LCT utilizes LLM-derived knowledge to aid in user weight calibration during the training phase of the recommender, thus mitigating the influence of fraudsters. Specifically, the LCT utilizes two types of knowledge: (1) specific knowledge that underscores the detailed differences between known fraudsters and genuine users within the recommender system, and (2) LLMs' general knowledge regarding the fraudulence potential of given user profiles. By combining current model feedback (as specific knowledge) with LLMs' general knowledge, the LCT gains both a zoomed-in view and a broad perspective for recognizing fraudsters in the current recommender system. As both the LCT and the sequential recommender system undergo ongoing optimization, the LCT continuously calibrates the training loss of the model via user-wise reweighing. This ensures a gradual reduction in the impact of fraudsters, preserving accurate and robust recommendations.

The pivotal contributions of our work are as follows:
\begin{itemize}[leftmargin=*]
    \item We pioneer the exploration of LLMs' knowledge of fraudsters within recommender systems, revealing how LLMs' knowledge can aid defense methods in generalizing across various attacks.
    \item We lead the initiative of incorporating LLMs into the robustness of sequential recommender systems, introducing LoRec as an innovative and general framework that employs LLM-enhanced Calibration for robust sequential recommendations.
    \item Our extensive experiments confirm the efficacy of the LoRec framework in withstanding diverse types of attacks and its adaptability across multiple backbone recommendation architectures.
\end{itemize}

\section{RELATED WORK}
\begin{figure*}
    \centering 
    \includegraphics[width=6.4in]{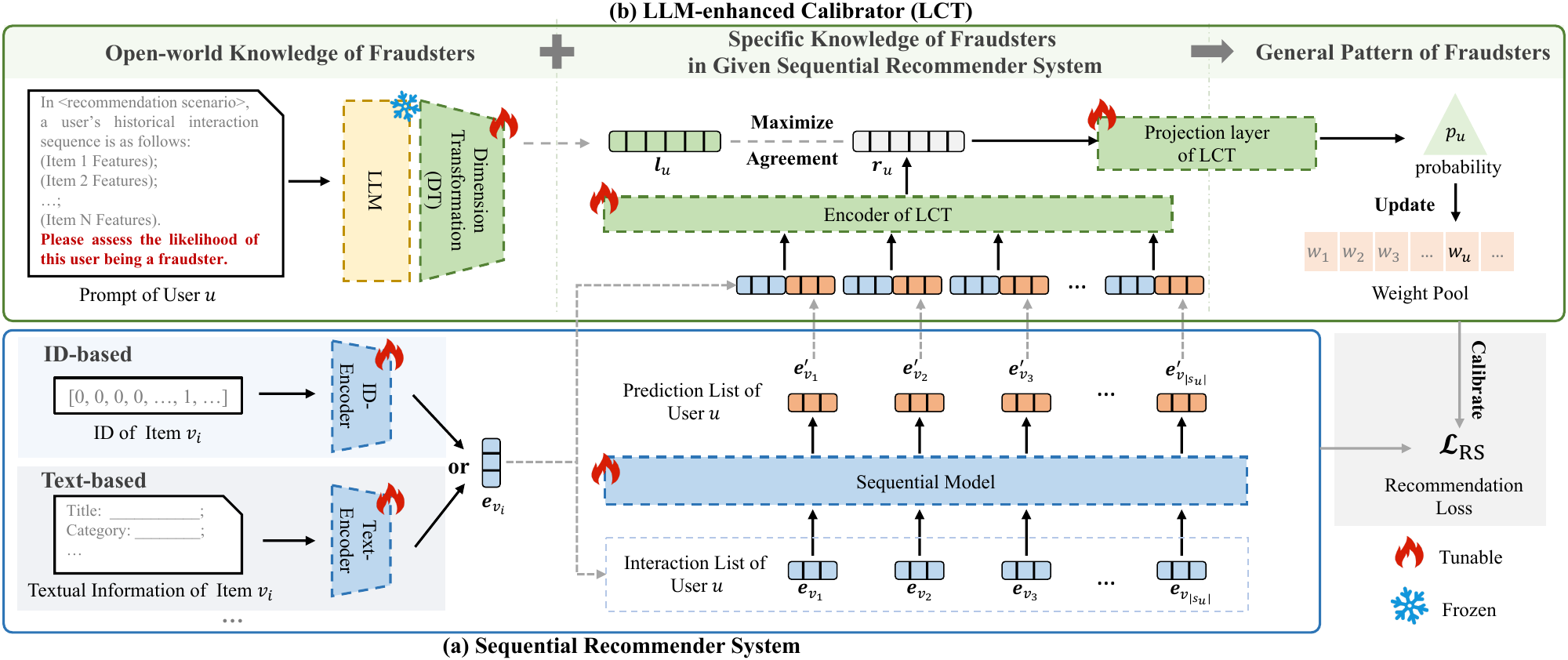}
    \caption{Overview of LoRec: (a) Sequential Recommender System with ID-based/Text-based item encoding; (b) LLM-enhanced Calibrator utilizes both open-world knowledge from LLMs and specific knowledge within the sequential recommender system to calibrate user weights, enhancing its robustness against poisoning attacks.}
    \label{fig:framework}
\end{figure*}

\subsection{Sequential Recommender System}
Early sequential recommender systems primarily rely on the Markov Chain framework to model user interactions~\cite{he2016fusing, he2016vista}. As neural network technologies advance, a shift occurs toward architectures like Recurrent Neural Networks~\cite{hidasi2015session, beutel2018latent} and Convolutional Neural Networks~\cite{tang2018personalized}. The introduction of sophisticated neural network architectures, particularly the Transformer~\cite{vaswani2017attention}, further enhances the capability of effectively capturing the nuances of dynamic user preferences, such as SASrec~\cite{kang2018selfattentive} and Bert4Rec~\cite{sun2019bert4rec}. More recently, researchers are exploring novel paradigms for modeling user interests, such as MLPs~\cite{zhou2022filterenhanced, zhang2023mlpst}, pushing the boundaries of the field.  Additionally, significant efforts are made to integrate self-supervised learning into the training of sequential recommender systems, aiming to enhance recommendation performance~\cite{zhou2020srec, xie2022contrastive, chen2022intent}. Despite these advancements, the susceptibility of these models to malicious attacks remains a concern, highlighting a significant challenge~\cite{zhang2023robust}.

\subsection{Robust Recommender System}
Mainstream methods for enhancing the robustness of recommender systems against poisoning attacks broadly fall into two categories~\cite{zhang2023robust}, the development of robust systems via (1) adversarial training~\cite{he2018adversarial, li2020adversarial, wu2021fight, yue2022defending, ye2023towards}, or (2) the removal of malicious data via detection techniques~\cite{chung2013beta, zhang2014hhtsvm, yang2016rescale, zhang2020gcnbased, liu2020recommending}.

Adversarial training, under the framework of Adversarial Personalized Ranking (APR)~\cite{he2018adversarial}, incorporates small adversarial perturbations at the parameter level during training~\cite{he2018adversarial, li2020adversarial, ye2023towards}. This process employs a ``min-max'' optimization strategy, aiming to minimize recommendation loss while maximizing the impact of adversarial perturbations. This approach is predicated on the assumption that attackers aim to degrade overall recommendation performance~\cite{he2018adversarial, wu2021fight}. 
However, real-world attacks often have varied objectives, like the promotion of specific items~\cite{huang2021data, tang2020revisiting}. Such attacks usually do not directly undermine recommendation performance, leading to a significant discrepancy with the ``min-max'' paradigm. Consequently, this diminishes the defensibility of adversarial training.

Detection techniques directly identify and remove potential fraudsters before training~\cite{chung2013beta, zhang2014hhtsvm, yang2016rescale, liu2020recommending}, or detect suspicious behaviors during training and softly discount their influence~\cite{zhang2020gcnbased}. These detection mechanisms often incorporate specific assumptions~\cite{chung2013beta, zhang2020gcnbased}, or rely on supervised attack data~\cite{zhang2014hhtsvm, yang2016rescale, zhang2020gcnbased}.
For example, GraphRfi~\cite{zhang2020gcnbased} assumes the unpredictability of fraudsters' behaviors, filtering them out by contrasting predictions with actual user behaviors.
Additionally, RAdabst~\cite{yang2016rescale} creates fraudulent data based on heuristic-based attacks~\cite{lam2004shilling, mobasher2007toward} to train its detection model. 
These assumptions or reliance can impede the ability of detection methods to adapt to changing attack patterns. 
For example, DP attacks~\cite{huang2021data}, which are optimization-based, employ poisoned model predictions to create fraudulent behaviors. This challenges GraphRfi's assumption of unpredictability~\cite{zhang2020gcnbased} and diverges from the supervised data used in RAdabst~\cite{yang2016rescale}.

In real-world scenarios, as attacks evolve, the ability of defense methods to generalize to unknown attacks is vital. However, there is still a shortage of effective defense strategies that can generalize across various types of unknown attacks.

\section{PRELIMINARY}
This section mathematically formulates the task of sequential recommendation.
We define the set of users as $\mathcal{U}$ and the set of items as $\mathcal{V}$. Each user $u \in \mathcal{U}$ is linked with a historical interaction sequence $s_u = [v_{1}, v_{2}, \dots, v_{|s_u|}]$, wherein $v_t \in \mathcal{V}$ denotes an item interacted with user $u$, and $|s_u|$ denotes sequence length. The objective of sequential recommendation is to predict the next item that user $u$ will interact with. This predictive task can be formulated as:
\begin{equation}
    v_{|s_u|+1} = \arg\max_{v \in \mathcal{V}} \mathbb{P}(v|s_u),
\end{equation}
where $\mathbb{P}(v|s_u)$ is the conditional probability of item $v$ being the next item following the historical interaction sequence $s_u$.

\section{Method}
\label{sec:method}

\subsection{Overview of LoRec}

Due to the openness of sequential recommender systems, attackers can inject fraudulent users into their training data. This manipulation disrupts the item transition patterns learned by these recommender systems, maliciously promoting target items. To effectively counter these malicious activities, our primary goal is to selectively reduce the influence of users identified as potential fraudsters, ensuring that the recommender systems remain robust and unaffected by such manipulations. 

A key component of this strategy is accurately differentiating between genuine users and fraudsters, especially in cases of previously unseen attacks.
Consequently, we explore integrating general knowledge from LLMs about the fraudulent potential of user historical interactions into specific knowledge highlighting the detailed differences between genuine users and known fraudsters in the recommender system, enhancing generalizability to unknown attacks.

Building on this foundation, we introduce LoRec, a general and adaptable framework suitable for a variety of sequential recommender systems. LoRec comprises two main components: a backbone sequential recommendation model and an LLM-enhanced CalibraTor (LCT), as illustrated in Figure~\ref{fig:framework}.

\textbf{Sequential Recommender System}. This component can be any existing sequential recommender system. It transforms users' historical interaction sequences into low-dimensional representations. These representations are utilized for predicting subsequent items with which users will interact, as depicted in Figure~\ref{fig:framework}(a).

\textbf{LLM-enhanced CalibraTor (LCT)}. 
The LCT, a novel component of LoRec, calibrates user weights by incorporating both specific knowledge within the current recommendation model and the open-world knowledge of LLMs. By combining these two types of knowledge, the LCT learns the general patterns of fraudsters. The LCT then employs these patterns to calibrate the weights assigned to each user during the training of the sequential recommender system, significantly enhancing the robustness against poisoning attacks, as shown in Figure~\ref{fig:framework}(b).

\subsection{Sequential Recommender System}

As recommender systems evolve, their models increasingly incorporate not only user interaction behaviors but also diverse item-related side information, including texts~\cite{hou2023learning,yuan2023where}, images~\cite{niu2018neural}, and other relevant data~\cite{deldjoo2020recommender}. 
Typically, for each item in the set $\mathcal{V}$, we utilize modality-specific encoders, e.g., ID-based encoders or Text-based encoders, to transform this side information into embeddings $\bm{e}_v \in \mathbb{R}^{d}$. This process is illustrated in the left part of Figure~\ref{fig:framework}(a).
The historical interaction sequence $s_u = [v_1, v_2, \dots, v_{|s_u|}]$ of user $u$ can be represented at embedding-level as:
\begin{equation}
    \bm{E}_u = [ \bm{e}_{v_1}, \bm{e}_{v_2}, \dots, \bm{e}_{v_{|s_u|}} ].
\end{equation}

The encoded sequence $\bm{E}_u$ serves as the input for the sequential model $g$ within any existing sequential recommender systems, which then generates the prediction embedding sequence $\bm{E}_u' = [ \bm{e}_{v_1}', \bm{e}_{v_2}', \dots, \bm{e}_{v_{|s_u|}}' ]$ for user $u$, where each $\bm{e}_{v_i}' \in \mathbb{R}^{d}$ is defined as: 
\begin{equation}
    \bm{e}_{v_i}' = g(\bm{E}_{u, 1:i}), \quad i = 1, 2, \dots, |s_u|.
\end{equation}

To forecast the likelihood of subsequent items in the sequence, we leverage both the item and prediction embeddings. The probability of item $v_j$ being the next in the sequence, given the user's historical interactions up to time $i$, is calculated as follows:
\begin{equation}
    \mathbb{P}\left(v_{i+1} = v_j|s_{u, 1:i} \right) = \sigma(\bm{e}_{v_j}^\top \bm{e}_{v_i}'),
\end{equation}
where $\sigma(\cdot)$ is the sigmoid function.

\subsection{LLM-enhanced Calibrator}
To effectively and accurately identify fraudsters in $\mathcal{U}$, the LCT predominantly leverages two knowledge types: specific knowledge from the current sequential recommender system regarding known fraudsters, and open-world knowledge derived from LLMs about whether a user appears fraudulent.
The specific knowledge underscores the detailed differences between genuine users and known fraudsters within the current sequential recommender system. Meanwhile, the open-world knowledge provides the LCT with a general understanding, enhancing its ability to generalize across different fraudster types. These comprehensive insights aid LCT in better identifying and diminishing the impact of users who are potentially harmful. By utilizing these two types of knowledge, the LCT gains both a zoomed-in view and a wide lens for recognizing the fraudsters in the current sequential recommender system.

\subsubsection{Specific Knowledge Modeling}
Specific knowledge refers to the distinctions between genuine users and known fraudsters within the current sequential recommender system. This knowledge offers detailed insights crucial for identifying fraudsters.
Considering the behavioral differences between genuine users and fraudsters, we leverage feedback from the current sequential recommender system to model the specific knowledge. This involves concatenating each item's embedding $\bm{e}_v$ (the feedback from the modality-specific encoder) with its predictive counterpart $\bm{e}_v'$ (the feedback from the sequential model), resulting in a composite embedding $\bm{k}_v \in \mathbb{R}^{2d}$ for each item in the historical interaction sequence as:
\begin{equation}
    \bm{k}_{v_i} = \mathrm{concat}(\bm{e}_{v_i}, \bm{e}_{v_i}'), \quad v_i \in s_u.
\end{equation}

Then, a learnable embedding $\bm{k}_0 \in \mathbb{R}^{2d}$ is used to summarize the features in $\bm{K}_u = [\bm{k}_{v_1}, \bm{k}_{v_2}, \dots, \bm{k}_{v_{|s_u|}}]$ through self-attention, producing $\bm{r}_u \in \mathbb{R}^{2d}$ as:
\begin{equation}
    \bm{r}_u = \mathrm{selfAtt}\left(\left[\bm{k}_0, \bm{K}_u\right]\right),
\end{equation}
where $\mathrm{selfAtt}(\cdot)$ is the self-attention mechanism with position embeddings as employed in \cite{vaswani2017attention}, and $\bm{r}_u$ is the output corresponding to $\bm{k}_0$.
Then, a two-layer perceptron, i.e., the projection layer $h$, is employed to map $\bm{r}_u$ onto $p_u \in [0, 1]$, which reflects the likelihood of ``user $u$ being a fraudster'':
\begin{equation}
    p_u = h(\bm{r}_u).
\end{equation}

\subsubsection{Open-world Knowledge Extraction}
LLMs incorporate vast open-world knowledge into their parameters, which can provide a more general understanding of identifying fraudsters. To leverage this knowledge, we transform user interaction data into prompts to effectively query the LLMs' knowledge about ``whether the given user is a fraudster''.
In more detail, we structure the recommendation context and the user's interaction sequence as:
\begin{adjustwidth}{4em}{4em}
\begin{quote}
    \textit{``In <\textbf{recommendation scenario}>, \\ a user's interaction sequence is as follows}: \\
    \textit{\textbf{Item 1 Features}; \\
    \textbf{Item 2 Features}; \\
    ...; \\
    \textbf{Item N Features}.}''
\end{quote}
\end{adjustwidth}
where ``<\textbf{recommendation scenario}>'' specifies the recommendation context, for example, news recommendation.
The term ``\textit{\textbf{Features}}'' is the item's side information, such as title and category.
We then pose an instruction to the LLMs:
\begin{adjustwidth}{1em}{1em}
\begin{quote}
    ``\textit{Please assess the likelihood of this user being a fraudster}.'' 
\end{quote}
\end{adjustwidth}
Let $\mathrm{prompt}_u$ denote the combination of the structured recommendation context, user's interaction sequence, and the instruction message of user $u$, as depicted in the left part of Figure~\ref{fig:framework}(b).

We transform the output of the LLM into a low-dimensional embedding as follows:
\begin{equation}
    \bm{l}_u = \mathrm{DT}(\mathrm{LLM}(\mathrm{prompt}_u)),
\end{equation}
where $\mathrm{DT}$ denotes a Dimension Transformation (DT) block, used for transforming the dimension of the embedding.
The obtained embedding $\bm{l}_u \in \mathbb{R}^{2d}$ contains the LLMs' open-world knowledge of whether user $u$ exhibits characteristics of fraudsters.

Next, we aim to maximize the agreement between the representations $\bm{l}_u$ and $\bm{r}_u$ to integrate the open-world knowledge:
\begin{equation}
    \max{\mathrm{sim}(\bm{l}_u, \bm{r}_u)}, \quad \forall u \in \mathcal{U},
    \label{eq:max_ag}
\end{equation}
where $\mathrm{sim}(\cdot)$ is the cosine similarity function. We can incorporate this agreement maximization into the LCT through loss optimization (see details in Section~\ref{sec:lct_train}). 

Such agreement maximization can further help the representation $\bm{r}_u$ considering fraudster characteristics from both specific knowledge within the current sequential recommender system and the LLMs’ open-world knowledge, thereby better reflecting the likelihood $p_u$ of the user $u$ being a fraudster.

\subsubsection{Training Strategy of LCT}
\label{sec:lct_train}
To train the LCT, we establish a set of known fraudulent users as the supervised dataset, $\mathcal{U}_{\mathrm{atk}}$, where $\mathcal{U}_{\mathrm{atk}} \bigcap \mathcal{U} =  \emptyset$. Users in $\mathcal{U}_{\mathrm{atk}}$ are assigned a label of 1.
Ideally, a set of verified genuine users would be assigned a label of 0. However, practically, we have access only to the observed user set $\mathcal{U}$, which might include some potential fraudsters.
As a compromise solution, we label users in $\mathcal{U}$ as 0 to train a model approximating the fraudster's likelihood as:
\begin{equation}
    \mathcal{L}_{\mathrm{F}} = - \frac{1}{|\mathcal{U}_\mathrm{atk}|} \sum_{u \in \mathcal{U}_\mathrm{atk}} \log(p_u) - \frac{1}{|\mathcal{U}|} \sum_{u \in \mathcal{U}} \log(1 - p_u).
    \label{eq:l-fd}
\end{equation}

\textbf{Regularization Loss}. Given that the user set $\mathcal{U}$ may include both genuine users and potential fraudsters~(injected by real attackers but we do not know), we implement an Entropy Regularization for users in $\mathcal{U}$ to avoid extreme predictions by the LCT:
\begin{equation}
    \mathcal{L}_{\mathrm{ER}} = \frac{1}{|\mathcal{U}|} \sum_{u \in \mathcal{U}} \left( \log(p_u) + \log(1 - p_u) \right).
\end{equation}

\textbf{LLM-enhanced Loss}. Lastly, we aim to integrate the open-world knowledge into $\bm{r}_u$ to extend the specific knowledge of known fraudsters into the general pattern within the current sequential recommender system. This is achieved by maximizing the agreement between $\bm{l}_u$ from the LLMs and $\bm{r}_u$ from the LCT encoder as per Equation~\ref{eq:max_ag}, utilizing cosine similarity as the measuring function:
\begin{equation}
    \mathcal{L}_{\mathrm{LLM}} = - \frac{1}{|\mathcal{U} \bigcup \mathcal{U}_{\mathrm{atk}}|}\sum_{u \in \mathcal{U} \bigcup \mathcal{U}_{\mathrm{atk}}} \mathrm{sim}(\bm{l}_u, \bm{r}_u).
\end{equation}

The final loss function for the LCT is thus defined as:
\begin{equation}
    \mathcal{L}_{\mathrm{LCT}} = \mathcal{L}_{\mathrm{F}} + \lambda_1 \mathcal{L}_{\mathrm{ER}} + \lambda_2 \mathcal{L}_{\mathrm{LLM}},
    \label{eq:LCT}
\end{equation}
where $\lambda_1$ and $\lambda_2$ are the hyperparameters.

\subsection{Calibration for Robust Recommendation}
As mentioned above, in practical scenarios, the existence of fraudsters within the training user set $\mathcal{U}$ is typically unknown. 
The user set $\mathcal{U}$ consists of both genuine users $\mathcal{U}_{n}$ and potential fraudsters $\mathcal{U}_{f}$.
However, all users, including those potentially fraudulent, are initially assigned a label of 0, as delineated in Equation~\ref{eq:l-fd}.
In such situations, it is challenging to guarantee a distinct demarcation between the $p_u$ values of genuine users and those of potential fraudsters, as shown in the bottom scenario of Figure~\ref{fig:frq}. To address this, we implement an adaptive threshold for more effectively identifying potential fraudsters and introduce an iterative weight compensation mechanism to minimize the misidentification risk of users.

\begin{figure}
    \centering
    \includegraphics[width=3.2in]{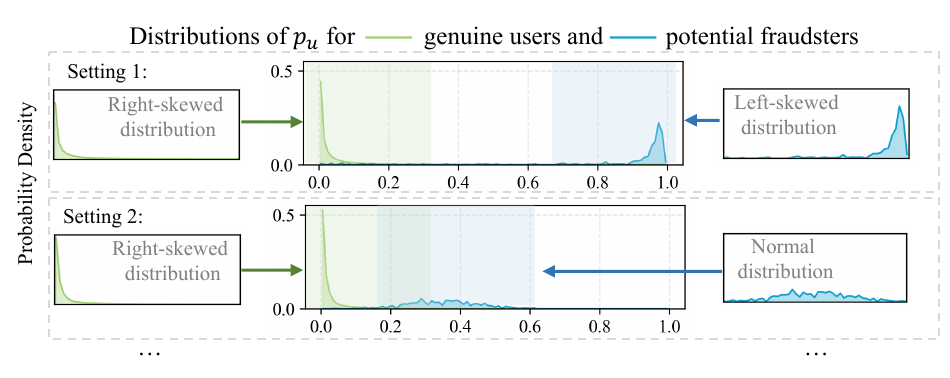}
    \caption{For LCT with different scenarios (hyperparameters, training epochs, or attack types),  the learned probability $p_u$ for genuine users  $\mathcal{U}_{n}$ and potential fraudsters $\mathcal{U}_{f}$ follow different distributions.}
    \label{fig:frq}
\end{figure}

\textbf{Observation}. A well-trained classifier often displays a skewed distribution in its logit outputs, a phenomenon sometimes referred to as over-confidence~\cite{pmlr-v162-wei22d}.
We leverage this characteristic for more effective fraudster identification. We observe that, for users in $\mathcal{U}_{n}$, the logits post-LCT projection layer exhibit a pronounced skew towards $p_u=0$, i.e., a right-skewed distribution\footnote{A right-skewed distribution is characterized by a long right tail, where the mean exceeds the median~\cite{doi:10.1080/10691898.2005.11910556}.}. For $\mathcal{U}_{f}$, the logit distribution can vary in different scenarios (different hyperparameters, training epochs, attacks, etc.), as shown in Figure~\ref{fig:frq}.

\textbf{Adaptive Threshold}. Without loss of generality, assuming a well-trained LCT, $p_u$ for $\mathcal{U}_{n}$ follow a right-skewed distribution with mean $\mu_{\mathrm{n}}$, while those for $\mathcal{U}_{f}$ approximate either a normal distribution or a left-skewed distribution with mean $\mu_{\mathrm{f}}$, where $\mu_{\mathrm{n}} < \mu_{\mathrm{f}}$. Let $\gamma$ represent the ratio $\frac{|\mathcal{U}_{f}|}{|\mathcal{U}_{n}|}$, satisfying $0 \leq \gamma < 1$. 
The mean $\mu_{\mathrm{o}}$ of the overall $p_u$ in $\mathcal{U}$ is:
\begin{equation}
    \mu_{\mathrm{o}} ~=~ \frac{1}{|\mathcal{U}|} \sum_{u \in \mathcal{U}} p_u ~=~ \frac{(\mu_{\mathrm{n}} + \gamma \mu_{\mathrm{f}})}{1+\gamma},
\end{equation}
where $\mu_{\mathrm{n}} < \mu_{\mathrm{o}} < \mu_{\mathrm{f}}$. Simply, we can employ $\mu_{\mathrm{o}}$  as the adaptive threshold to identify potential fraudsters.

\textbf{Iterative Weight Compensation}. However, directly applying $\mu_{\mathrm{o}}$ as the adaptive threshold introduces two challenges: (1) some genuine users with $p_u > \mu_{\mathrm{o}}$ may be misclassified, while (2) some potential fraudsters remain unidentified. To mitigate this, we introduce an iterative weight compensation mechanism.
We establish a weight pool tracking per-user weight coefficients $\xi_u, \forall u \in \mathcal{U}$, initially set to $\hat{\xi}$. Upon reaching certain training epochs with the LCT, $\xi_u$ is updated as:
\begin{equation}
    \xi_u(t+1) = \left\{ 
        \begin{aligned}
            \xi_u(t) - 1, \quad & \text{if } p_u > \mu_\mathrm{o}, \\
            \xi_u(t) + \frac{\sum_{u \in \mathcal{U}} \mathbb{I}\left( p_u > \mu_\mathrm{o} \right)}{\sum_{u \in \mathcal{U}} \mathbb{I}\left( p_u \leq \mu_\mathrm{o} \right)}, \quad & \text{if } p_u \leq \mu_\mathrm{o},
        \end{aligned}
    \right.
    \label{eq:update}
\end{equation}
where $\mathbb{I}(\cdot)$ returns 1 when the condition is true. The notation $\xi_u(t)$ represents the value of $\xi_u$ after $t$ updates, with $\xi_u(0) = \hat{\xi}$.

\begin{proposition}
    Consider i.i.d. samples $n_0, n_1, \dots, n_{N}$ drawn from a right-skewed distribution with mean $\mu_{\mathrm{n}}$, and $z_0, z_1, \dots, z_{Z}$ drawn from either a normal distribution or a left-skewed distribution with mean $\mu_{\mathrm{f}}$. Let $\xi_{n_i}$ and $\xi_{z_j}$ represent the weight coefficients for $n_i$ and $z_j$, respectively. We denote $\overline{\xi_{n}} = \frac{1}{N}\sum_{i=0}^N \xi_{n_i}$, and $\overline{\xi_{z}} = \frac{1}{Z}\sum_{j=0}^Z \xi_{z_j}$. Given $\mu_{\mathrm{o}}$ satisfying $\mu_{\mathrm{n}} < \mu_{\mathrm{o}} < \mu_{\mathrm{f}}$, and updating these coefficients according to Equation~\ref{eq:update}, it can be shown that $\mathbb{E}[\overline{\xi_{n}}(t) - \overline{\xi_{n}}(t+1)] > 0$ and $\mathbb{E}[\overline{\xi_{z}}(t) - \overline{\xi_{z}}(t+1)] < 0$.
    \label{pro:xi}
\end{proposition}

\begin{proof}
    Let $\alpha$ denote the probability $\mathbb{P}_n(n > \mu_{\mathrm{o}})$ and $\beta$ denote the probability $\mathbb{P}_z(z > \mu_{\mathrm{o}})$. Let $\gamma$ represent the ratio $\frac{Z}{N}$. 
    Define $\mathrm{F}_{n}(\cdot)$ as the Cumulative Distribution Function~(CDF) of the right-skewed distribution $\mathbb{P}_n(n)$, characterized by mean $\mu_\mathrm{n}$ and median $m_\mathrm{n}$, and $\mathrm{F}_{z}(\cdot)$ as the CDF of either the normal or left-skewed distribution $\mathbb{P}_z(z)$, with mean $\mu_\mathrm{f}$ and median $m_\mathrm{f}$.
    It follows that $\mathrm{F}_{n}(\mu_\mathrm{n}) > \mathrm{F}_{n}(m_\mathrm{n}) = 0.5$, and $\mathrm{F}_{z}(\mu_\mathrm{f}) \leq \mathrm{F}_{z}(m_\mathrm{f}) = 0.5$.
    With $\mu_\mathrm{n} < \mu_\mathrm{o} < \mu_\mathrm{f}$, we derive:
    \begin{equation}
        \begin{aligned}
            \alpha &= 1 - \mathrm{F}_{n}(\mu_\mathrm{o}) < 1 - \mathrm{F}_{n}(\mu_\mathrm{n}) < 0.5, \\
            \beta &= 1 - \mathrm{F}_{z}(\mu_\mathrm{o}) > 1 - \mathrm{F}_{z}(\mu_\mathrm{f}) \geq 0.5.
        \end{aligned}
    \end{equation}
    Therefore, $\alpha < \beta$.
    The change between $\overline{\xi_{z}}(t)$ and $\overline{\xi_{z}}(t+1)$ at each update is given by
    \begin{equation}
        \begin{aligned}
            \mathbb{E}[\Delta \overline{\xi}_z] &= \left[\frac{\alpha+\gamma\beta}{1+\gamma - (\alpha+\gamma\beta)} \cdot \gamma N(1-\beta) - \gamma N \beta\right] \cdot (\gamma N)^{-1} \\
            &= \frac{\alpha - \beta}{1+\gamma - (\alpha + \gamma \beta)}.
        \end{aligned}
    \end{equation}

    Since $\alpha < \beta$ and $\alpha, \beta \in (0,1)$, it follows that $\mathbb{E}[\Delta \overline{\xi}_z] < 0$. Similarly, $\mathbb{E}[\Delta \overline{\xi}_n] > 0$ for $\xi_n$. Thus, Proposition~\ref{pro:xi} is proved.
\end{proof}

Continuous training of the sequential recommender system and consequent in changes of LCT inputs for each $u \in \mathcal{U}$ lead to \textbf{evolving and rearranging} $p_u$ distributions from LCT outputs. Based on Proposition~\ref{pro:xi}, this ensures that the weight coefficients for fraudsters decrease with each update.

\textbf{Training Calibration}. We then calculate user weights as:
\begin{equation}
    w_u = q\cdot\sigma(\xi_u),
\end{equation}
where $q = \sigma(\hat{\xi})^{-1}$ so that all the weights used in the initial state are 1.0. Finally, the loss of sequential recommender system~(using SASrec~\cite{kang2018selfattentive} as an example) is:
\begin{equation}
    \mathcal{L}_{\mathrm{RS}} = - \sum_{u \in \mathcal{U}}\! w_u \!\sum_{i=1}^{|s_u|-1} \Bigl[ \log \Bigl(  \sigma \bigl( \bm{e}_{v_{i+1}}^\top \bm{e}_{v_i}' \bigr) \Bigr) {+}\!\!  \sum_{j \notin s_{u,1:i+1}}\!\!\!\! \log \Bigl( 1 - \sigma \bigl(\bm{e}_{v_j}^\top \bm{e}_{v_i}' \bigr) \Bigr)\Bigr].
\end{equation}

\section{EXPERIMENTS}
In this section, we conduct extensive experiments to answer the following research questions (\textbf{RQs}).
\begin{itemize}[leftmargin=*]
    \item \textbf{RQ1:} Whether LoRec can defend against poisoning attacks?
    \item \textbf{RQ2:} What does each component of LoRec bring?
    \item \textbf{RQ3:} Is LoRec adaptable to diverse recommendation settings and backbone models?
\end{itemize}

\begin{table}[t]
  \centering
    \caption{Dataset statistics}
    \resizebox{0.47\textwidth}{!}{
        \begin{threeparttable}

\begin{tabular}{lrrrrr}
    \toprule
    \textbf{ DATASET } & \textbf{ \#Users } & \textbf{ \#Items } & \textbf{\#Ratings}  & \textbf{Avg.length} & \textbf{Sparsity}\\
    \midrule
     Games  & 61,521 & 33,243 & 541,789 & 8.8 & 99.97\% \\
     Arts  & 71,364 & 61,505 & 600,989 & 8.4 & 99.99\% \\ 
     MIND  & 152,909 & 63,608 & 4,186,679 & 27.4 & 99.96\% \\ 
    \bottomrule
    \end{tabular}
        \end{threeparttable}
    }
  \label{tab:datasets}%
\end{table}%
\subsection{Experimental Setup}

\begin{table*}[t]
    \centering
    \caption{Robustness against target items promotion}
    \resizebox{\textwidth}{!}{
        \begin{threeparttable}

\begin{tabular}{clcccccccc}
    \toprule
    \multicolumn{1}{c}{\multirow{2}{*}{\textbf{Dataset}}}& \multicolumn{1}{c}{\multirow{2}{*}{\textbf{Model}}} & \multicolumn{2}{c}{\textbf{Random Attack}(\%)} & \multicolumn{2}{c}{\textbf{Bandwagon Attack}(\%)} & \multicolumn{2}{c}{\textbf{DP Attack}(\%)} & \multicolumn{2}{c}{\textbf{Rev Attack}(\%)} \\ 
    \cmidrule(lr){3-4} \cmidrule(lr){5-6} \cmidrule(lr){7-8} \cmidrule(lr){9-10}
    & & \textbf{T-HR@50}\tnote{1} & \textbf{T-NDCG@50} & \textbf{T-HR@50} & \textbf{T-NDCG@50} & \textbf{T-HR@50} & \textbf{T-NDCG@50} & \textbf{T-HR@50} & \textbf{T-NDCG@50} \\
    \midrule
    \multirow{9}{1.2cm}{\centering \textbf{Games}} & \textbf{Backbone} & 0.889 $\pm$ 0.073 & 0.242 $\pm$ 0.006 & 0.904 $\pm$ 0.138 & 0.232 $\pm$ 0.010 & 0.458 $\pm$ 0.070 & 0.113 $\pm$ 0.005 & 0.858 $\pm$ 0.154 & 0.235 $\pm$ 0.014\\
    & ~~+\textbf{StDenoise} & 0.633 $\pm$ 0.029 & 0.174 $\pm$ 0.003 & 1.106 $\pm$ 0.150 & 0.288 $\pm$ 0.011 & 0.334 $\pm$ 0.026 & 0.079 $\pm$ 0.002 & 1.132 $\pm$ 0.136 & 0.310 $\pm$ 0.011 \\
    & ~~+\textbf{CL4Srec} & 0.748 $\pm$ 0.025 & 0.199 $\pm$ 0.002 & 1.165 $\pm$ 0.104 & 0.302 $\pm$ 0.009 & 0.529 $\pm$ 0.064 & 0.129 $\pm$ 0.005 & 1.240 $\pm$ 0.145 & 0.346 $\pm$ 0.012 \\
    & ~~+\textbf{APR} & 0.377 $\pm$ 0.047 & 0.162 $\pm$ 0.012 & 0.756 $\pm$ 0.056 & 0.224 $\pm$ 0.005 & 0.449 $\pm$ 0.069 & 0.118 $\pm$ 0.006 & 0.362 $\pm$ 0.002 & 0.126 $\pm$ 0.000 \\
    & ~~+\textbf{ADVTrain} & 0.962 $\pm$ 0.065 & 0.294 $\pm$ 0.008 & 1.170 $\pm$ 0.017 & 0.305 $\pm$ 0.001 & 0.336 $\pm$ 0.046 & 0.082 $\pm$ 0.003 & 0.713 $\pm$ 0.088 & 0.210 $\pm$ 0.010 \\
    & ~~+\textbf{GrapRfi} & 0.819 $\pm$ 0.037 & 0.225 $\pm$ 0.003 & 0.895 $\pm$ 0.075 & 0.231 $\pm$ 0.006 & 0.506 $\pm$ 0.024 & 0.122 $\pm$ 0.001 & 0.950 $\pm$ 0.137 & 0.267 $\pm$ 0.013\\
    \cmidrule{2-10} 
    & ~~+\textbf{LLM4Dec} & 0.303 $\pm$ 0.009 & 0.078 $\pm$ 0.001 & 0.235 $\pm$ 0.006 & 0.057 $\pm$ 0.000 & 0.319 $\pm$ 0.008 & 0.077 $\pm$ 0.001 & 0.432 $\pm$ 0.020 & 0.112 $\pm$ 0.001\\
    & ~~+\textbf{LoRec} & \textbf{0.068 $\pm$ 0.002} & \textbf{0.016 $\pm$ 0.000} & \textbf{0.105 $\pm$ 0.007} & \textbf{0.024 $\pm$ 0.000} & \textbf{0.103 $\pm$ 0.001} & \textbf{0.024 $\pm$ 0.000} & \textbf{0.080 $\pm$ 0.001} & \textbf{0.019 $\pm$ 0.000}\\
    \cmidrule{3-10}
    & \multicolumn{1}{c}{Gain\tnote{2}} & +81.97\% $\uparrow$ & +89.89\% $\uparrow$& +86.16\% $\uparrow$ & +89.10\% $\uparrow$& +69.04\% $\uparrow$ & +69.61\% $\uparrow$& +78.05\% $\uparrow$ & +84.73\% $\uparrow$\\
    \midrule
    \multirow{9}{1.2cm}{\centering \textbf{Arts}} 
    & \textbf{Backbone} &5.646 $\pm$ 1.030 & 1.926 $\pm$ 0.298 & 4.078 $\pm$ 1.168 & 1.109 $\pm$ 0.109 & 1.978 $\pm$ 0.529 & 0.479 $\pm$ 0.044 & { \footnotesize OOM}\tnote{3} & { \footnotesize OOM}\\
    & ~~+\textbf{StDenoise} & 4.498 $\pm$ 0.979 & 1.312 $\pm$ 0.100 & 4.822 $\pm$ 0.327 & 1.340 $\pm$ 0.028 & 2.195 $\pm$ 0.974 & 0.611 $\pm$ 0.090 & { \footnotesize OOM} & { \footnotesize OOM}\\
    & ~~+\textbf{CL4Srec} & 4.988 $\pm$ 0.926 & 1.479 $\pm$ 0.119 & 4.517 $\pm$ 0.710 & 1.282 $\pm$ 0.080 & 1.676 $\pm$ 0.320 & 0.420 $\pm$ 0.024 & { \footnotesize OOM} & { \footnotesize OOM} \\
    & ~~+\textbf{APR} & 5.331 $\pm$ 0.696 & 1.467 $\pm$ 0.464 & 3.762 $\pm$ 0.619 & 1.077 $\pm$ 0.443 & 1.943 $\pm$ 0.128 & 0.917 $\pm$ 0.010 & { \footnotesize OOM} & { \footnotesize OOM}\\
    & ~~+\textbf{ADVTrain} & 3.520 $\pm$ 0.927 & 1.009 $\pm$ 0.089 & 4.659 $\pm$ 3.614 & 1.316 $\pm$ 0.370 & 1.886 $\pm$ 0.338 & 0.504 $\pm$ 0.031 & { \footnotesize OOM} & { \footnotesize OOM} \\
    & ~~+\textbf{GrapRfi} & 5.331 $\pm$ 0.609 & 1.553 $\pm$ 0.048 & 3.542 $\pm$ 1.544 & 0.957 $\pm$ 0.120 & 1.814 $\pm$ 0.408 & 0.452 $\pm$ 0.025 & { \footnotesize OOM} & { \footnotesize OOM}\\
    \cmidrule{2-10} 
    & ~~+\textbf{LLM4Dec} & 1.338 $\pm$ 0.194 & 0.342 $\pm$ 0.015 & 0.679 $\pm$ 0.021 & 0.176 $\pm$ 0.002 & 0.790 $\pm$ 0.015 & 0.197 $\pm$ 0.001 & { \footnotesize OOM} & { \footnotesize OOM}\\
    & ~~+\textbf{LoRec} & \textbf{0.576 $\pm$ 0.027} & \textbf{0.141 $\pm$ 0.002} & \textbf{0.436 $\pm$ 0.061} & \textbf{0.108 $\pm$ 0.004} & \textbf{0.440 $\pm$ 0.042} & \textbf{0.109 $\pm$ 0.003} & { \footnotesize OOM} & { \footnotesize OOM}\\
    \cmidrule{3-10}
    & \multicolumn{1}{c}{Gain} & +42.42\% $\uparrow$ & +85.89\% $\uparrow$& +56.44\% $\uparrow$ & +88.72\% $\uparrow$& +56.02\% $\uparrow$ & +74.08\% $\uparrow$& - & - \\
    \midrule
    \multirow{9}{1.2cm}{\centering \textbf{MIND}} 
    & \textbf{Backbone} & 0.215 $\pm$ 0.015 & 0.073 $\pm$ 0.002 & 0.259 $\pm$ 0.006 & 0.083 $\pm$ 0.001 & 0.099 $\pm$ 0.002 & 0.026 $\pm$ 0.002 & { \footnotesize OOM} & { \footnotesize OOM}\\
    & ~~+\textbf{StDenoise} & 0.193 $\pm$ 0.001 & 0.062 $\pm$ 0.001 & 0.337 $\pm$ 0.028 & 0.111 $\pm$ 0.003 & 0.088 $\pm$ 0.003 & 0.025 $\pm$ 0.002 & { \footnotesize OOM} & { \footnotesize OOM}\\
    & ~~+\textbf{CL4SRec} & 0.166 $\pm$ 0.008 & 0.054 $\pm$ 0.001 & 0.259 $\pm$ 0.015 & 0.082 $\pm$ 0.002 & 0.093 $\pm$ 0.001 & 0.026 $\pm$ 0.000 & { \footnotesize OOM} & { \footnotesize OOM} \\
    & ~~+\textbf{APR} & 0.135 $\pm$ 0.020 & 0.047 $\pm$ 0.001 & 0.256 $\pm$ 0.004 & 0.088 $\pm$ 0.001 & 0.082 $\pm$ 0.002 & 0.024 $\pm$ 0.001 & { \footnotesize OOM} & { \footnotesize OOM}\\
    & ~~+\textbf{ADVTrain} & 0.141 $\pm$ 0.001 & 0.048 $\pm$ 0.000 & 0.319 $\pm$ 0.010 & 0.104 $\pm$ 0.001 & 0.127 $\pm$ 0.003 & 0.036 $\pm$ 0.000 & { \footnotesize OOM} & { \footnotesize OOM} \\
    & ~~+\textbf{GrapRfi} & 0.118 $\pm$ 0.001 & 0.038 $\pm$ 0.001 & 0.206 $\pm$ 0.003 & 0.067 $\pm$ 0.001 & 0.123 $\pm$ 0.003 & 0.034 $\pm$ 0.001 & { \footnotesize OOM} & { \footnotesize OOM}\\
    \cmidrule{2-10} 
    & ~~+\textbf{LLM4Dec} & 0.044 $\pm$ 0.001 & 0.013 $\pm$ 0.000 & 0.219 $\pm$ 0.003 & 0.068 $\pm$ 0.002 & 0.051 $\pm$ 0.000 & 0.014 $\pm$ 0.002 & { \footnotesize OOM} & { \footnotesize OOM}\\
    & ~~+\textbf{LoRec} & \textbf{0.005 $\pm$ 0.001} & \textbf{0.001 $\pm$ 0.000} & \textbf{0.012 $\pm$ 0.001} & \textbf{0.003 $\pm$ 0.001} & \textbf{0.004 $\pm$ 0.001} & \textbf{0.001 $\pm$ 0.000} & { \footnotesize OOM} & { \footnotesize OOM}\\
    \cmidrule{3-10}
    & \multicolumn{1}{c}{Gain} 
    & +95.61\% $\uparrow$ & +96.28\% $\uparrow$& +94.29\% $\uparrow$ & +95.10\% $\uparrow$& +95.07\% $\uparrow$ & +95.92\% $\uparrow$& - & -\\
    \bottomrule
\end{tabular}
\begin{tablenotes}
    \item[1] Target Item Hit Ratio (Equation~\ref{eq:tar}); T-HR@50 and T-NDCG@50 of all target items on clean datasets are 0.000.
    \item[2] The relative percentage increase of LoRec's metrics to the best value of other baselines' metrics, i.e., $\left(\min\left(\mathrm{T}\text{-}\mathrm{HR}_\mathrm{Beslines}\right) - \mathrm{T}\text{-}\mathrm{HR}_\mathrm{LoRec} \right)/ \min(\mathrm{T}\text{-}\mathrm{HR}_\mathrm{Beslines})$.
    \item[3] The Rev attack method could not be executed on the dataset due to memory constraints, resulting in an out-of-memory error.
\end{tablenotes}
   
        \end{threeparttable}
    }
\label{tab:attack_per}%
\end{table*}

\subsubsection{Datasets}
In our evaluation of LoRec, we employ three widely recognized datasets: the Amazon review datasets (\textbf{Games} and \textbf{Arts})~\cite{ni2019justifying}, and the \textbf{MIND} news recommendation dataset~\cite{wu2020mind}. For the Amazon datasets, all users are considered. For the MIND dataset, a subset of users is sampled following~\cite{li2023exploring}. Consistent with existing practices~\cite{rendle2012bpr, kang2018selfattentive}, we exclude users with fewer than 5 interactions. For each user, the data split process involves (1) using the most recent action for testing, (2) the second most recent action for validation, and (3) all preceding actions (up to 50) for training~\cite{kang2018selfattentive, li2023exploring, yuan2023where}. Dataset statistics are provided in Table~\ref{tab:datasets}.

\subsubsection{Backbone Models}
We employ three backbone sequential recommender systems:
\begin{itemize}[leftmargin=*]
    \item \textbf{GRU4rec}~\cite{hidasi2015session} utilizes Recurrent Neural Networks~\cite{cho2014learning} to model user interaction sequences in session-based recommendations.
    \item \textbf{SASrec}~\cite{kang2018selfattentive} employs a multi-head self-attention mechanism in Transformer~\cite{vaswani2017attention} for sequential recommendations.
    \item \textbf{FMLPrec}~\cite{zhou2022filterenhanced} is an all-Multilayer Perceptron model with a learnable filter-enhanced block for noise reduction in embeddings for sequential recommendations.
\end{itemize}
\textbf{Due to space limitations}, we predominantly show results using SASrec as the backbone model in Section~\ref{sec:performance}. Results for GRU4rec and FMLPrec are described in Section~\ref{sec:diff_rec}. Additionally, we consider both Text-based and ID-based sequential recommender systems as depicted in Figure~\ref{fig:framework}(a), but mainly focus on Text-based version, while ID-based results are elaborated in Section~\ref{sec:diff_set}.

\subsubsection{Baselines for Defense}
We include various methods such as the adversarial training methods APR~\cite{he2018adversarial} and ADVTrain~\cite{yue2022defending}; a detection-based method, GraphRfi~\cite{zhang2020gcnbased}. We also include two denoise-based approaches, StDenoise~\cite{tian2022learning, ye2023towards} and CL4Srec~\cite{xie2022contrastive}. 
\begin{itemize}[leftmargin=*]
    \item \textbf{APR}~\cite{he2018adversarial}: Adversarial training in recommender systems generates small parameter perturbations and integrates these perturbations into training.
    \item \textbf{ADVTrain}~\cite{yue2022defending}: This adversarial training technique counters profile pollution attacks in sequential recommender systems.
    \item \textbf{GraphRfi}~\cite{zhang2020gcnbased}: A Graph Convolutional Network and Neural Random Forest-based framework for fraudster detection during training. We adapt it in sequential recommendations.
    \item \textbf{StDenoise}~\cite{tian2022learning, ye2023towards}: A structural denoising method that leveraging similarity between $\bm{e}_v$ and $\bm{e}_v'$ for $v \in s_u$ as evidence for noise purification, as employed in \cite{tian2022learning, ye2023towards}.
    \item \textbf{CL4Srec}~\cite{xie2022contrastive}: A contrastive learning framework for sequential recommender systems to resist noise. We specifically use the ``Crop'' data augmentation technique.
\end{itemize}
Additionally, we present a simplified version of LoRec as a baseline, named \textbf{LLM4Dec}, which focuses solely on detection using LLM-encoded user behavior $\bm{l}_u$, i.e., $p_u = h(\bm{l}_u)$.

\subsubsection{Attack Methods}
We consider heuristic attacks, i.e., Random Attack~\cite{lam2004shilling} and Bandwagon Attack~\cite{mobasher2007toward}, alongside optimized-based methods, i.e., DP Attack~\cite{huang2021data} and Rev Attack~\cite{tang2020revisiting}. These attacks are executed in a black-box scenario where the attackers lack knowledge about the victim model’s architecture and parameters.
\begin{itemize}[leftmargin=*]
    \item \textbf{Random Attack} (Heuristic)~\cite{lam2004shilling}: Interactions of fraudsters include both target items and randomly selected items.
    \item \textbf{Bandwagon Attack} (Heuristic)~\cite{mobasher2007toward}: Interactions of fraudsters include target items and items selected based on their popularity.
    \item \textbf{DP Attack} (Optimization-based)~\cite{huang2021data}: Targeting deep learning recommender systems specifically.
    \item \textbf{Rev Attack} (Optimization-based)~\cite{tang2020revisiting}: Framing attacks as bi-level optimization problem solved by gradient-based techniques.
\end{itemize}

\subsubsection{Evaluation Metrics}
For recommendation performance, the primary metrics are the top-$k$ recommendation performance metrics: Hit Ratio at $k$ ($\mathrm{HR}@k$) and Normalized Discounted Cumulative Gain at $k$ ($\mathrm{NDCG}@k$), as referenced in \cite{zhang2023robust, kang2018selfattentive}. To assess the success ratio of attacks, we use metrics specific to the target items' top-$k$ performance, denoted as $\mathrm{T}\text{-}\mathrm{HR}@k$ and $\mathrm{T}\text{-}\mathrm{NDCG}@k$~\cite{tang2020revisiting, huang2021data}:
\begin{equation}
    \mathrm{T}\text{-}\mathrm{HR}@k = \frac{1}{|\mathcal{T}|} \sum_{\mathit{tar} \in \mathcal{T}} \frac{ \sum_{u \in \mathcal{U}_{n} - \mathcal{U}_{n, \mathit{tar}}} \mathbb{I}\left(\mathit{tar} \in L_{u_{1:k}}\right)}{|\mathcal{U}_{n} - \mathcal{U}_{n, \mathit{tar}}|},
    \label{eq:tar}
\end{equation}
where $\mathcal{T}$ represents the set of target items, $\mathcal{U}_{n, \mathit{tar}}$ denotes the set of genuine users who interacted with target item $\mathit{tar}, L_{u_{1:k}}$ is the top-$k$ recommendation list for user $u$, and $\mathbb{I}(\cdot)$ is an indicator function that returns 1 if the condition is true. $\mathrm{T}\text{-}\mathrm{NDCG}@k$ mirrors $\mathrm{T}\text{-}\mathrm{HR}@k$, serving as the target item-specific version of $\mathrm{NDCG}@k$.

Additionally, in line with the robustness definition presented in \cite{zhang2023robust}, we consider the variance in recommendation performance before and after an attack as an indicator of robustness, termed top-$k$ Recommendation Consistency ($\mathrm{RC}@k$):
\begin{equation}
    \mathrm{RC}_{\mathrm{HR}}@k = 1 - \frac{|\mathrm{HR}@k - \mathrm{HR}@k_{\mathrm{clean}}|}{\mathrm{HR}@k_{\mathrm{clean}}},
    \label{eq:rc}
\end{equation}
where $\mathrm{HR}@k_{\mathrm{clean}}$ is the top-$k$ HR on clean dataset. $\mathrm{RC}_{\mathrm{NDCG}}@k$ is calculated similarly for the $\mathrm{NDCG}@k$.

\begin{table*}[t]
    \centering
    \caption{Recommendation consistency}
    \resizebox{\textwidth}{!}{
        \begin{threeparttable}

\begin{tabular}{lcccccccc}
    \toprule
    \multicolumn{1}{c}{\multirow{2}{*}{\textbf{Model}}} & \multicolumn{2}{c}{\textbf{Clean} (\%)} & \multicolumn{2}{c}{\textbf{Random Attack} (\%)} & \multicolumn{2}{c}{\textbf{Bandwagon Attack} (\%)} & \multicolumn{2}{c}{\textbf{DP Attack} (\%)}  \\ 
    \cmidrule(lr){2-3} \cmidrule(lr){4-5} \cmidrule(lr){6-7} \cmidrule(lr){8-9}
    & \textbf{RC$_{\bm{\mathrm{HR}}}$@10 (HR)\tnote{1}} & \textbf{RC$_{\bm{\mathrm{NDCG}}}$@10 (NDCG)} & \textbf{RC$_{\bm{\mathrm{HR}}}$@10 (HR)} & \textbf{RC$_{\bm{\mathrm{NDCG}}}$@10 (NDCG)} & \textbf{RC$_{\bm{\mathrm{HR}}}$@10 (HR)} & \textbf{RC$_{\bm{\mathrm{NDCG}}}$@10 (NDCG)} & \textbf{RC$_{\bm{\mathrm{HR}}}$@10 (HR)} & \textbf{RC$_{\bm{\mathrm{NDCG}}}$@10 (NDCG)}  \\
    \midrule
     \textbf{Backbone} & - (11.289) & - (5.628) & 99.46\% (11.228) & 99.73\% (5.613) & 98.52\% (11.122) & 98.44\% (5.540) & 98.64\% (11.442) & 98.26\% (5.726) \\
     \cmidrule{1-9}
     ~~+\textbf{StDenoise} & 98.72\% (11.144) & 98.79\% (5.560) & 99.68\% (11.253) & 99.63\% (5.607) & \textbf{99.97\%} (11.286) & 98.79\% (5.696) & 99.98\% (11.287) & \textbf{99.96\%} (5.630) \\
     ~~+\textbf{CL4Srec} & 99.17\% (11.383) & 99.89\% (5.567) & 97.67\% (11.026) & 97.85\% (5.507) & 99.03\% (11.179) & 99.00\% (5.572) & 99.87\% (11.304) & 99.88\% (5.635) \\
     ~~+\textbf{APR} & 90.81\% (10.252) & 92.11\% (5.184) & 94.06\% (10.618) & 97.21\% (5.471) & 90.18\% (10.180) & 91.33\% (5.140) & 89.09\% (10.057) & 91.58\% (5.154) \\
     ~~+\textbf{ADVTrain} & 99.15\% (11.193) & 96.09\% (5.408) & 99.47\% (11.229) & 96.57\% (5.435) & 98.89\% (11.163) & 96.57\% (5.435) & 98.35\% (11.475) & 97.97\% (5.514) \\
     ~~+\textbf{GraphRfi} & 99.28\% (11.370) & 98.60\% (5.707) & 98.41\% (11.469) & 98.53\% (5.711) & 99.34\% (11.215) & 98.97\% (5.570) & 98.86\% (11.418) & 97.73\% (5.756) \\
     \cmidrule{1-9}
     ~~+\textbf{LLM4Dec} & 98.90\% (11.165) & 98.56\% (5.547) & 99.49\% (11.231) & 99.63\% (5.607) & 96.68\% (10.914) & 96.64\% (5.439) & 98.20\% (11.086) & 98.56\% (5.547) \\
     ~~+\textbf{LoRec} & \textbf{99.66\%} (11.327) & \textbf{99.00\%} (5.684) & \textbf{99.96\%} (11.293) & \textbf{99.80\%} (5.639) & 99.06\% (11.183) & \textbf{99.40\%} (5.594) & \textbf{99.98\%} (11.287) & 99.34\% (5.665) \\

    \bottomrule

\end{tabular}
\begin{tablenotes}
    \item[1] Recommendation Consistency (Equation~\ref{eq:rc}). Here, HR and NDCG are abbreviations of HR@10 and NDCG@10, respectively.
\end{tablenotes}
   
        \end{threeparttable}
    }
\label{tab:attack_rc}%
\end{table*}

\subsubsection{Implementation Details} 
We adopt \textbf{Llama2} family into LCT. For evaluating the success ratio of attacks, we set $k=50$~\cite{huang2021data, tang2020revisiting, wu2021fight}. For recommendation performance metrics, we utilize $k=10$~\cite{kang2018selfattentive, zhou2022filterenhanced}.
For both defense methods and backbone models, the learning rate is selected from \{0.1, 0.01, $\dots, 1 \times 10^{-5}$\}. Similarly, weight decay is chosen from \{0, 0.1, $\dots, 1 \times 10^{-5}$\}.
Backbone model architectures follow their original publications. For the APR, we select the parameter $\epsilon$ from within the range of \{0.02, 0.03, $\dots, 0.05$\}. This selection is based on the observation that values of $\epsilon$ greater than 0.05 tend to negatively impact recommendation performance. The implementation of GraphRfi follows its paper. For the detection-based methods and our method, the Bandwagon Attack is provided as supervised fraudster data, denoted as $\mathcal{U}_{\mathrm{atk}}$. The size of $\mathcal{U}_{\mathrm{atk}}$ is consistently set at $10\%$ of $|\mathcal{U}|$. Regarding the attack methods, the attack budget is established at $1\%$ with five target items. The hyperparameters are in alignment with those detailed in their original publications. Our implementation code is accessible via the provided link\footnote{\url{https://github.com/Kaike-Zhang/LoRec}}.

\subsection{Performance Comparison~(RQ1)}
\label{sec:performance}

In this section, we answer \textbf{RQ1}. We focus on two key aspects: its robustness against poisoning attacks and its robustness in maintaining recommendation performance.

\subsubsection{Robustness Against Poisoning Attacks}
We evaluate the efficacy of LoRec in defending against poisoning attacks, focusing on the success ratio of attacks. In our experiments, we specifically target extremely unpopular items, resulting in values of $\mathrm{T}\text{-}\mathrm{HR}@50$ and $\mathrm{T}\text{-}\mathrm{NDCG}@50$ being 0.0 without any attack. 
\textbf{Note}: The lower the values of $\mathrm{T}\text{-}\mathrm{HR}@50$ and $\mathrm{T}\text{-}\mathrm{NDCG}@50$, the better the defensibility.
Table~\ref{tab:attack_per} shows that denoise-based methods, such as StDenoise and CL4Srec, exhibit unstable performance against attacks, sometimes even increasing the attack success ratio, suggesting the inadequacy of simple denoising in attack defense. GraphRfi generally exhibits better defense against attacks similar to its supervised data. However, with different attack types, such as DP attacks and Rev Attacks, GraphRfi's effectiveness significantly diminishes, potentially leading to an increased success ratio of attacks.

In contrast, LLM4Dec, depending exclusively on LLMs' knowledge for detection, surpasses most baselines. This highlights the capability of LLMs' knowledge in identifying fraudsters.
Furthermore, LoRec demonstrates superior performance compared to the baselines, significantly lowering the success ratio of attacks. For the three datasets, LoRec reduces the average $\mathrm{T}\text{-}\mathrm{HR}@50$ by 78.81\%, 51.63\%, and 94.99\%, and the average $\mathrm{T}\text{-}\mathrm{NDCG}@50$ by 83.33\%, 82.90\%, and 95.76\%, respectively, compared to the best results of the baselines. These results demonstrate the effectiveness of LoRec in defending against various known or unknown attacks.

\subsubsection{Robustness in Maintaining Recommendation Performance}
We evaluate LoRec's ability to preserve recommendation consistency in the face of poisoning attacks. Due to space limitations, we present results only for the MIND dataset. As shown in Table~\ref{tab:attack_rc}, LoRec surpasses or is comparable with existing defense mechanisms, achieving an average $\mathrm{RC}_{\mathrm{HR}}@10$ of 99.67\% and $\mathrm{RC}_{\mathrm{NDCG}}@10$ of 99.39\%.

\begin{figure}
    \centering
    \includegraphics[width=3in]{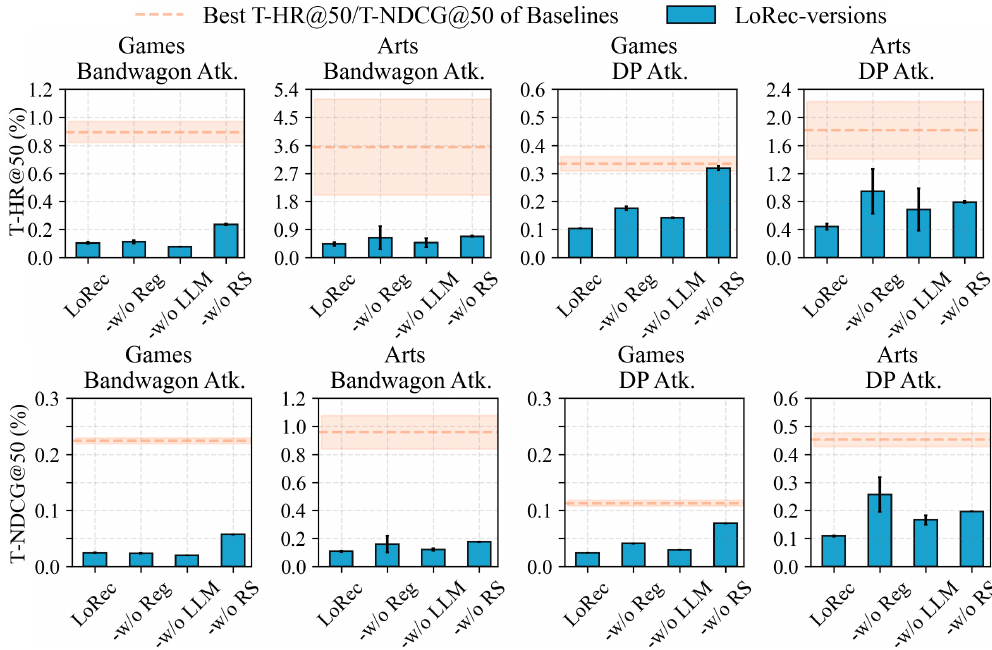}
    \caption{Ablation study}
    \label{fig:ab_study}
\end{figure}

\subsection{Augmentation Analysis~(RQ2)}

\begin{figure}
    \centering
    \includegraphics[width=3in]{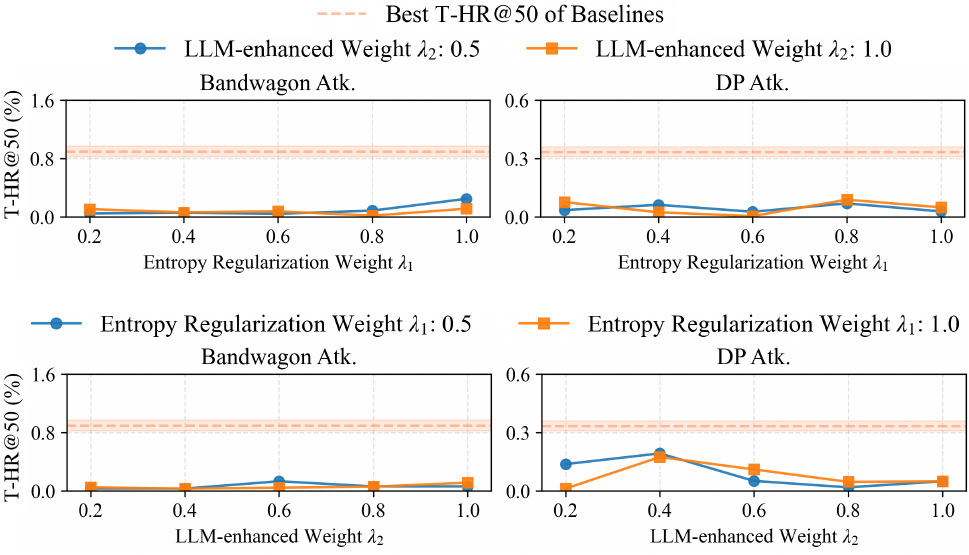}
    \caption{Hyperparameter analysis on Entropy Regularization $\lambda_1$ and LLM-enhanced weights $\lambda_2$}
    \label{fig:para_anal}
\end{figure}

\begin{figure}
    \centering
    \includegraphics[width=3in]{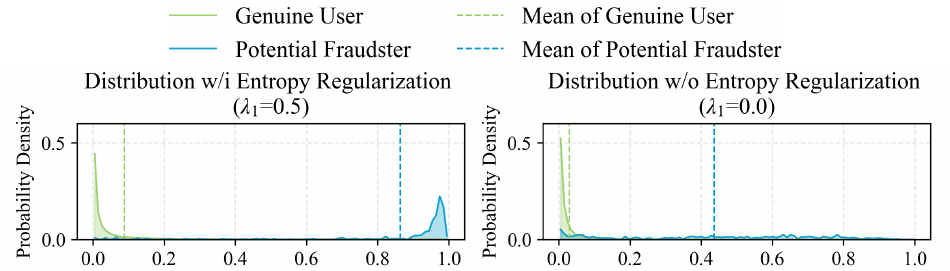}
    \caption{Influence of Entropy Regularization on distributions of genuine users and potential fraudsters}
    \label{fig:entropy}
\end{figure}

In this section, we answer \textbf{RQ2}. Our focus is on analyzing the efficacy of each component within LoRec and examining the impact of hyperparameters and the size of the LLM. Due to space limitations, we only show the results on Games and Arts, while the MIND dataset holds consistent results.

\subsubsection{Ablation Study}
We first examine different variants of LoRec to evaluate the contribution of each component to its overall performance. The variants are defined as follows:
\begin{itemize}[leftmargin=*]
    \item \textbf{-w/o Reg}: This version operates without Entropy Regularization, setting $\lambda_1 = 0$ in Equation~\ref{eq:LCT}.
    \item \textbf{-w/o LLM}: This version operates without the open-world knowledge provided by LLMs, setting $\lambda_2 = 0$ in Equation~\ref{eq:LCT}.
    \item \textbf{-w/o RS}: This version operates without feedback from sequential recommender systems, rendering it equivalent to LLM4Dec.
\end{itemize}

As depicted in Figure~\ref{fig:ab_study}, our findings reveal that each component plays a vital role in ensuring robust recommendations. All variants of LoRec outperform the best metrics achieved by the baseline models. Specifically, Entropy Regularization stabilizes the defense capability, reducing the variance of the performance. Additionally, it is observed that the LLMs' open-world knowledge has a more pronounced impact when countering the unknown DP attacks, as compared to the Bandwagon Attack which is the same type of supervised data. LLMs' knowledge aids LoRec in deriving a more general pattern beyond the supervised data, enhancing its effectiveness in tackling other unseen attacks. Lastly, feedback from the sequential recommender system emerges as a crucial element. This feedback significantly enhances LoRec's ability to accurately mitigate fraudsters' impact for achieving robust recommendations.

\subsubsection{Hyperparameter Analysis}
We explore the effects of hyperparameters, i.e., Entropy Regularization Weight $\lambda_1$ and LLM-Enhanced Weight $\lambda_2$ as defined in Equation~\ref{eq:LCT}. We investigate the influence of varying $\lambda_2$ while holding $\lambda_1$ constant at either 0.5 or 1.0, and conversely, altering $\lambda_1$ while $\lambda_2$ remains fixed at 0.5 or 1.0. As shown in Figure~\ref{fig:para_anal}, the results indicate that the performance remains comparatively stable despite changes in $\lambda_1$ or $\lambda_2$. 

As illustrated in the lower part of Figure~\ref{fig:para_anal}, adjustments in $\lambda_2$ reveal that, for the Bandwagon Attack (mirroring the supervised data), a smaller $\lambda_2$ is sufficient for LoRec's effective defense. This sufficiency stems from the supervised data providing ample information for defending against similar attacks. Conversely, for the DP attacks (differing from the supervised data), a larger $\lambda_2$ enhances LoRec's performance. This improvement is attributed to the additional useful open-world knowledge, which assists the model in learning more general patterns of attacks.

Moreover, we investigate the influence of Entropy Regularization Weight $\lambda_1$ on the distributions of $p_u$. By preventing the model from making extreme predictions, Entropy Regularization contributes to more distinctly separating the $p_u$ of potential fraudsters from those of genuine users, as illustrated in Figure~\ref{fig:entropy}.

\subsubsection{LLMs' Size Analysis}
We investigate the influence of LLMs in varying sizes on LoRec's performance. As shown in Figure~\ref{fig:LLM_version}, we find that as the size of the LLMs increases, there is a general trend towards enhanced defensibility against poisoning attacks in LoRec.

\begin{figure}
    \centering
    \includegraphics[width=3.2in]{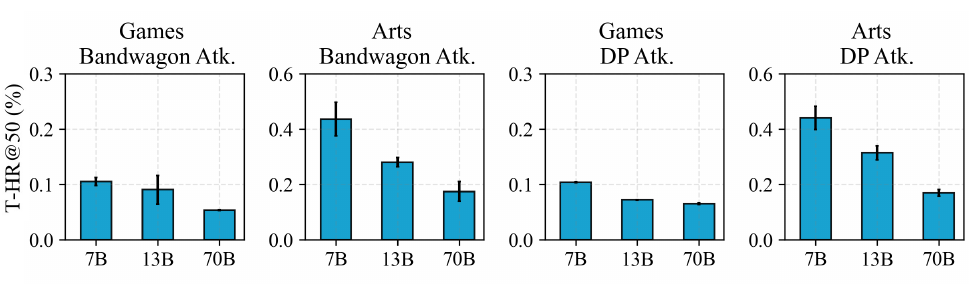}
    \caption{Performance comparison of LoRec variants utilizing Llama2 with different sizes}
    \label{fig:LLM_version}
\end{figure}

\begin{figure}
    \centering
    \includegraphics[width=3.2in]{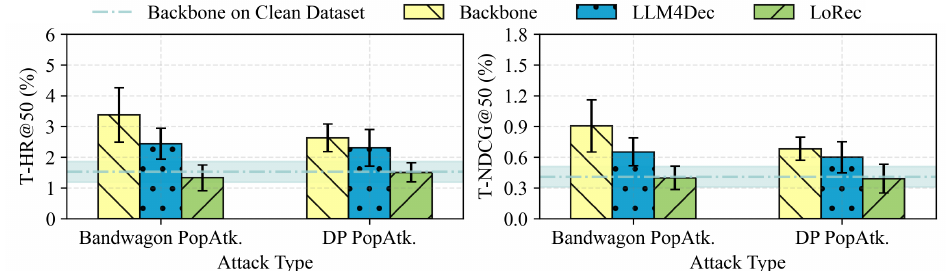}
    \caption{Robustness against the popular items promotion}
    \label{fig:pop_item}
\end{figure}

\subsection{Generality of LoRec~(RQ3)}

We answer \textbf{RQ3} by evaluating LoRec's performance in diverse settings and across various backbone sequential recommender systems. 
Due to space constraints and the fact that LLM4Dec surpasses all baseline methods (refers to Table~\ref{tab:attack_per}), we present results focusing on the Games and Arts datasets with LLM4Dec and LoRec.

\subsubsection{Defensibility against Different Types of Target Item}
Here, we evaluate the defensive capabilities against attacks targeting popular items, as presented in Figure~\ref{fig:pop_item}. These findings reveal that LoRec maintains a robust defense even against attacks aimed at promoting popular items, closely aligning with the backbone model's $\mathrm{T}\text{-}\mathrm{HR}@50$ and $\mathrm{T}\text{-}\mathrm{NDCG}@50$ metrics on clean datasets.

\subsubsection{Adaptation to Different Settings in Sequential Recommender Systems}
\label{sec:diff_set}
We investigate the adaptability of LoRec across different configurations to assess its generalizability. In the ID-based setup, LoRec consistently delivers the most robust outcomes, confirming its superior performance, as detailed in Table~\ref{tab:rec_id}.

\subsubsection{Generalization Across Various Recommender Systems}
\label{sec:diff_rec}
To ascertain LoRec's broad applicability, we evaluate its performance using different backbone sequential recommendation models. As indicated in Table~\ref{tab:rec_others}, we employ GRU4rec~\cite{hidasi2015session} and FMLPrec~\cite{zhou2022filterenhanced} as alternative backbone sequential recommender systems. It is noted that FMLPrec itself demonstrates the strongest robustness against poisoning attacks among SASrec (the backbone model in Table~\ref{tab:attack_per}), GRU4rec, and FMLPrec. For all these backbone models, LoRec consistently achieves the most robust results against attacks.

\begin{table}[t]
    \centering
    \caption{Robustness against target items promotion of ID-based sequential recommender system}
    \resizebox{0.47\textwidth}{!}{
        \begin{threeparttable}

\begin{tabular}{clcccc}
    \toprule
    \multicolumn{1}{c}{\multirow{2}{*}{\textbf{Dataset}}}& \multicolumn{1}{c}{\multirow{2}{*}{\textbf{Model}}} & \multicolumn{2}{c}{\textbf{Bandwagon Attack}(\%)} & \multicolumn{2}{c}{\textbf{DP Attack}(\%)}\\ 
    \cmidrule(lr){3-4} \cmidrule(lr){5-6}
    & & \textbf{T-HR@50} & \textbf{T-NDCG@50} & \textbf{T-HR@50} & \textbf{T-NDCG@50} \\
    \midrule
    \multirow{3}{1.2cm}{\centering \textbf{Games}} & \textbf{Backbone} & 1.517 $\pm$ 0.677 & 0.416 $\pm$ 0.053 & 0.275 $\pm$ 0.043 & 0.061 $\pm$ 0.002\\
    & ~~+\textbf{LLM4Dec} & 0.499 $\pm$ 0.185 & 0.130 $\pm$ 0.014 & 0.237 $\pm$ 0.054 & 0.052 $\pm$ 0.003\\
    & ~~+\textbf{LoRec} & \textbf{0.292 $\pm$ 0.166} & \textbf{0.080 $\pm$ 0.014} & \textbf{0.086 $\pm$ 0.015} & \textbf{0.020 $\pm$ 0.001}\\
    \midrule
    \multirow{3}{1.2cm}{\centering \textbf{Arts} } & \textbf{Backbone} & 7.098 $\pm$ 3.853 & 2.603 $\pm$ 0.240 & 1.708 $\pm$ 0.631 & 0.427 $\pm$ 0.038\\
    & ~~+\textbf{LLM4Dec} & 1.261 $\pm$ 0.034 & 0.366 $\pm$ 0.002 & 1.092 $\pm$ 0.003 & 0.320 $\pm$ 0.000\\
    & ~~+\textbf{LoRec} & \textbf{0.609 $\pm$ 0.507} & \textbf{0.162 $\pm$ 0.035} & \textbf{0.224 $\pm$ 0.016} & \textbf{0.051 $\pm$ 0.001}\\
    \bottomrule

\end{tabular}
   
        \end{threeparttable}
    }
\label{tab:rec_id}%
\end{table}

\begin{table}[t]
    \centering
    \caption{Robustness against target items promotion of GRU4rec and FMLPrec}
    \resizebox{0.47\textwidth}{!}{
        \begin{threeparttable}

\begin{tabular}{clcccc}
    \toprule
    \multicolumn{1}{c}{\multirow{2}{*}{\textbf{Dataset}}}& \multicolumn{1}{c}{\multirow{2}{*}{\textbf{Model}}} & \multicolumn{2}{c}{\textbf{Bandwagon Attack}(\%)} & \multicolumn{2}{c}{\textbf{DP Attack}(\%)}\\ 
    \cmidrule(lr){3-4} \cmidrule(lr){5-6}
    & & \textbf{T-HR@50} & \textbf{T-NDCG@50} & \textbf{T-HR@50} & \textbf{T-NDCG@50} \\
    \midrule
    \multirow{6}{1.2cm}{\centering \textbf{Games} } & \textbf{GRU4rec} & 1.623 $\pm$ 0.143 & 0.462 $\pm$ 0.014 & 1.141 $\pm$ 0.460 & 0.310 $\pm$ 0.046 \\
    & ~~+\textbf{LLM4Dec} & 0.605 $\pm$ 0.168 & 0.249 $\pm$ 0.053 & 0.992 $\pm$ 0.115 & 0.212 $\pm$ 0.020 \\
    & ~~+\textbf{LoRec} & \textbf{0.229 $\pm$ 0.010} & \textbf{0.061 $\pm$ 0.001} & \textbf{0.144 $\pm$ 0.010} & \textbf{0.044 $\pm$ 0.002} \\
    \cmidrule{2-6}
    & \textbf{FMLPrec} & 0.228 $\pm$ 0.122 & 0.072 $\pm$ 0.014 & 0.143 $\pm$ 0.002 & 0.036 $\pm$ 0.000 \\
    & ~~+\textbf{LLM4Dec} & 0.138 $\pm$ 0.011 & 0.029 $\pm$ 0.000 & 0.104 $\pm$ 0.013 & 0.027 $\pm$ 0.001 \\
    & ~~+\textbf{LoRec} & \textbf{0.038 $\pm$ 0.001} & \textbf{0.009 $\pm$ 0.000} & \textbf{0.071 $\pm$ 0.011} & \textbf{0.010 $\pm$ 0.001} \\
    \midrule
    \multirow{6}{1.2cm}{\centering \textbf{Arts}} & \textbf{GRU4rec} & 7.008 $\pm$ 5.912 & 2.445 $\pm$ 0.982 & 3.573 $\pm$ 1.769 & 1.149 $\pm$ 0.283 \\
    & ~~+\textbf{LLM4Dec} & 2.135 $\pm$ 0.954 & 0.765 $\pm$ 0.262 & 1.754 $\pm$ 0.420 & 0.493 $\pm$ 0.045 \\
    & ~~+\textbf{LoRec} & \textbf{0.358 $\pm$ 0.057} & \textbf{0.091 $\pm$ 0.004} & \textbf{0.601 $\pm$ 0.204} & \textbf{0.177 $\pm$ 0.020} \\
    \cmidrule{2-6}
    & \textbf{FMLPrec} & 0.675 $\pm$ 0.089 & 0.179 $\pm$ 0.006 & 1.471 $\pm$ 0.474 & 0.480 $\pm$ 0.071 \\
    & ~~+\textbf{LLM4Dec} & 0.108 $\pm$ 0.000 & 0.027 $\pm$ 0.000 & 0.436 $\pm$ 0.053 & 0.130 $\pm$ 0.005 \\
    & ~~+\textbf{LoRec} & \textbf{0.096 $\pm$ 0.003} & \textbf{0.024 $\pm$ 0.000} & \textbf{0.328 $\pm$ 0.016} & \textbf{0.098 $\pm$ 0.001} \\
    \bottomrule

\end{tabular}
   
        \end{threeparttable}
    }
\label{tab:rec_others}%
\end{table}

\section{CONCLUSION}

In this work, we introduce \textbf{LoRec}, a novel framework that employs LLM-enhanced Calibration to enhance the robustness of sequential recommender systems against poisoning attacks. LoRec integrates an LLM-enhanced Calibrator, integrating extensive open-world knowledge from LLMs into specific knowledge from the current recommender system to estimate the likelihood of users being fraudsters. These probabilities are then applied to calibrate the weight of each user during the training phase of the sequential recommender system. Through a process of iterative refinement, LoRec effectively diminishes the impact of fraudsters. Our extensive experimental analysis shows that LoRec, as a general framework, enhances the robustness of sequential recommender systems against poisoning attacks while also preserving recommendation performance. 

This study primarily focuses on the scenario of sequential recommendations. When expanding the proposed LoRec to other recommendation scenarios like collaborative filtering, additional challenges arise from accurately modeling user feedback while avoiding the impact of supervised fraudster data on the recommender system. We look forward to further research and future work to enhance the applicability of the proposed framework to a wider range of recommendation scenarios, thereby advancing the field.

\section{Acknowledgments}
This work is funded by the National Key R\&D Program of China (2022YFB3103700, 2022YFB3103701), the Strategic Priority Research Program of the Chinese Academy of Sciences under Grant No. XDB0680101, and the National Natural Science Foundation of China under Grant Nos. 62102402, 62272125, U21B2046. Huawei Shen is also supported by Beijing Academy of Artificial Intelligence (BAAI).

\bibliographystyle{ACM-Reference-Format}
\bibliography{ref}

\appendix

\end{document}